\documentclass[preprint]{elsarticle}

\usepackage{times,amssymb,amsmath,amsfonts,float,nicefrac,color,mathrsfs,cite,epic,eepic,stmaryrd}
\usepackage{euscript,graphics}
\usepackage[dvips]{epsfig}
\usepackage{graphicx, color}

\evensidemargin=\oddsidemargin
\newtheorem{theorem}{Theorem}[section]
\newtheorem{lemma}[theorem]{Lemma}
\newtheorem{proposition}[theorem]{Proposition}

\newtheorem{definition}{Definition}

\newcommand{\remove}[1]{}
\renewcommand{\tilde}{\widetilde}

\newproof{proof}{Proof}
\newcommand\nd{\noindent}%\newcommand\bysame{\rule[1mm]{1cm}{.025cm}}

\def\supp{\qopname\relax{no}{supp}}

\def\rank{\qopname\relax{no}{rk}}
\def\shape{\qopname\relax{no}{shape}}
\newcommand\nc\newcommand
\nc\bfa{{\boldsymbol a}}\nc\bfA{{\bf A}}\nc\cA{{\mathcal A}}
\nc\bfb{{\boldsymbol b}}\nc\bfB{{\bf B}}\nc\cB{{\mathcal B}}
\nc\bfc{{\boldsymbol c}}\nc\bfC{{\bf C}}\nc\cC{{\mathcal C}}
\nc\bfd{{\boldsymbol d}}\nc\bfD{{\bf D}}\nc\cD{{\mathcal D}}
\nc\bfe{{\boldsymbol e}}\nc\bfE{{\bf E}}\nc\cE{{\mathcal E}}
\nc\bff{{\boldsymbol f}}\nc\bfF{{\bf F}}\nc\cF{{\mathcal F}}
\nc\bfg{{\boldsymbol g}}\nc\bfG{{\bf G}}\nc\cG{{\mathcal G}}
\nc\bfh{{\boldsymbol h}}\nc\bfH{{\bf H}}\nc\cH{{\mathcal H}}
\nc\bfi{{\boldsymbol i}}\nc\bfI{{\bf I}}\nc\cI{{\mathcal I}}
\nc\bfj{{\boldsymbol j}}\nc\bfJ{{\bf J}}\nc\cJ{{\mathcal J}}
\nc\bfk{{\boldsymbol k}}\nc\bfK{{\bf K}}\nc\cK{{\mathcal K}}
\nc\bfl{{\boldsymbol l}}\nc\bfL{{\bf L}}\nc\cL{{\mathcal L}}
\nc\bfm{{\boldsymbol m}}\nc\bfM{{\bf M}}\nc\cM{{\mathcal M}}
\nc\bfn{{\boldsymbol n}}\nc\bfN{{\bf N}}\nc\cN{{\mathcal N}}
\nc\bfo{{\boldsymbol o}}\nc\bfO{{\bf O}}\nc\cO{{\mathcal O}}
\nc\bfp{{\boldsymbol p}}\nc\bfP{{\bf P}}\nc\cP{{\mathcal P}}
\nc\bfq{{\boldsymbol q}}\nc\bfQ{{\bf Q}}\nc\cQ{{\mathcal Q}}
\nc\bfr{{\boldsymbol r}}\nc\bfR{{\bf R}}\nc\cR{{\mathcal R}}
\nc\bfs{{\boldsymbol s}}\nc\bfS{{\bf S}}\nc\cS{{\mathcal S}}
\nc\bft{{\boldsymbol t}}\nc\bfT{{\bf T}}\nc\cT{{\mathcal T}}
\nc\bfu{{\boldsymbol u}}\nc\bfU{{\bf U}}\nc\cU{{\mathcal U}}
\nc\bfv{{\boldsymbol v}}\nc\bfV{{\bf V}}\nc\cV{{\mathcal V}}
\nc\bfw{{\boldsymbol w}}\nc\bfW{{\bf W}}\nc\cW{{\mathcal W}}
\nc\bfx{{\boldsymbol x}}\nc\bfX{{\bf X}}\nc\cX{{\mathcal X}}
\nc\bfy{{\boldsymbol y}}\nc\bfY{{\bf Y}}\nc\cY{{\mathcal Y}}
\nc\bfz{{\boldsymbol z}}\nc\bfZ{{\bf Z}}\nc\cZ{{\mathcal Z}}
\nc\od{{\bar d}}\nc\ow{{\bar w}}\nc\odelta{{\bar\delta}}
\nc\ox{{\bar x}}\nc\oy{{\bar y}}\nc\ou{{\bar u}}
\nc\oh{{\bar h}}

\newcommand\complexes{{\mathbb C}}
\newcommand\ff{{\mathbb F}}

\newcommand\integers{{\mathbb Z}}
\newcommand\n{{\llbracket n \rrbracket}}
\newcommand\fs[1]{{\llbracket #1 \rrbracket}}
\nc\ellone{{\ell_1}}
\nc\elltwo{{\ell_2}}
\nc\ellinf{{{\ell_\infty}}}
\nc\ideal[1]{\langle #1\rangle}

\newcommand{\beeq}{\begin{eqnarray*}}
\newcommand{\eneq}{\end{eqnarray*}}

\journal{Discrete Mathematics}

\begin{document}

\begin{frontmatter}

%% Title, authors and addresses

%% use the tnoteref command within \title for footnotes;
%% use the tnotetext command for the associated footnote;
%% use the fnref command within \author or \address for footnotes;
%% use the fntext command for the associated footnote;
%% use the corref command within \author for corresponding author footnotes;
%% use the cortext command for the associated footnote;
%% use the ead command for the email address,
%% and the form \ead[url] for the home page:
%%
%% \title{Title\tnoteref{label1}}
%% \tnotetext[label1]{}
%% \author{Name\corref{cor1}\fnref{label2}}
%% \ead{email address}
%% \ead[url]{home page}
%% \fntext[label2]{}
%% \cortext[cor1]{}
%% \address{Address\fnref{label3}}
%% \fntext[label3]{}

\title{Linear codes on posets with extension property}

%% use optional labels to link authors explicitly to addresses:
%% \author[label1,label2]{<author name>}
%% \address[label1]{<address>}
%% \address[label2]{<address>}

\author[A. Barg]{Alexander Barg\fnref{fn1}}
\ead{abarg@umd.edu}

\author[L. Felix,M. Firer]{Luciano V. Felix}
\ead{luvifelix@ufrrj.br}

\author[M. Firer]{Marcelo Firer\fnref{fn2}}
\ead{mfirer@ime.unicamp.br}

\author[M. Firer,M. Spreafico]{Marcos V.P. Spreafico}
\ead{marcos.spreafico@ufms.br}

\fntext[fn1]{Research supported in part by NSA under grant H98230-12-1-0260 and NSF grants
CCF0916919, CCF1217245, CCF1217894, DMS1101697.}
\fntext[fn2]{Research partially supported by FAPESP under grants 2007/ 56052-8 and 2102/20181-7.}

\address[A. Barg]{Dept. of Electrical and Computer Engineering and Institute for Systems Research, University of Maryland,
College Park, MD 20742, USA, and Institute for Problems of Information Transmission, Russian Academy of
Sciences, Moscow, Russia.}

\address[L. Felix]{ICE -- UFRRJ, Universidade Federal Rural do Rio de Janeiro, BR 465, Km 7, 23890-000 -- Serop\'{e}dica -- RJ, Brazil.}

\address[M. Firer]{IMECC -- UNICAMP, Universidade Estadual de Campinas, Rua S{\'e}rgio Buarque de Holanda 651, 13083-859 -- Campinas -- SP, Brazil.}

\address[M. Spreafico]{INMA -- UFMS, Universidade Federal de Mato Grosso do Sul, Av. Costa e Silva, 79070-900 -- Campo Grande -- MS, Brazil.}

\begin{abstract}
We investigate linear and additive codes in partially ordered Hamming-like spaces
that satisfy the extension property, meaning that automorphisms of ideals extend to
automorphisms of the poset. 
The codes are naturally described in terms of translation association schemes
that originate from the groups of linear isometries of the space. We address questions
of duality and invariants of codes, establishing a connection between the dual
association scheme and the scheme defined on the dual poset (they are isomorphic
if and only if the poset is self-dual). We further discuss invariants that play
the role of weight enumerators of codes in the poset case. 
In the case of regular rooted trees such invariants are linked to the 
classical problem of tree isomorphism. We also study the question of whether
these invariants are preserved under standard operations on posets such as the
ordinal sum and the like.

\end{abstract}

\begin{keyword}
Poset codes \sep Association schemes \sep MacWilliams relations 

%% keywords here, in the form: keyword \sep keyword

%% MSC codes here, in the form: \MSC code \sep code
%% or \MSC[2008] code \sep code (2000 is the default)

\end{keyword}

\end{frontmatter}

%%
%% Start line numbering here if you want
%%
% \linenumbers

%% main text
\section{Introduction}
%\label{}

The theory of linear codes is classically developed in the Hamming space over a finite field. Algebraic
aspects of this theory are connected with the theory of association schemes, one of the main thrusts
of which is related to duality theory of schemes and codes. While many results in this framework
extend to additive codes over group alphabets, linear codes continue to be the main object of study.
The focus of this work is poset metric spaces, i.e., finite spaces in which the distance is derived
from partial orders coordinates. Poset metric spaces
were introduced by Brualdi et al.~\citep{bru95} following the work of Niederreiter on one special case of this problem \citep{nie91}. Extension of coding theory to poset metric  spaces has been the subject of numerous publications in  the last decade \citep{skr01,lee03,kim05b,kim05a,bar09b,pan09b}.

%In this paper we address two aspects of the theory of linear codes on posets.
Many basic theorems for linear codes are related to the notion of code duality which itself is
derived from duality of the underlying association schemes. In this paper we are interested in
duality of association schemes that arise from groups of linear isometries of poset metric spaces.
Apart from being a convenient tool for the study of code duality \citep{del73,bro89,mar99,hyu06},
group actions give rise to invariants of linear codes that are used in the study of structural 
and extremal properties of codes. These studies are particularly interesting when the underlying 
association scheme is self-dual. We show that group actions on poset metric spaces give rise
to self-dual schemes if and only if the poset itself is self-dual. This proof relies on
the structure of the isometry group of the poset metric space that was recently established by Panek et al.~\citep{pan09b}.
Another property that we involve is extension of automorphisms from order ideals to the entire poset;
if this is possible, we say that the poset has the extension property.

Examples of posets studied in the literature are mostly confined to the hierarchical poset and
the ordered Hamming space, defined below in the paper. Looking outside this set of examples,
we consider posets whose Hasse diagrams are given by (level-)regular rooted trees. The poset weight is 
preserved under automorphisms, but this invariant is not sufficiently refined to characterize the orbits. Identifying such invariants leads us to the classical problem of
encoding of rooted trees \citep{Reed72} and deciding isomorphism of trees \citep{aho74}. Following
\citep{mar99,bar09b}, we call such invariants shapes of codevectors. Finally, we consider standard operations to build new posets out of given ones, and we analyze the behavior of the extension property under those operations, making explicit the behavior of a shape when the extension property may be ensured. We remark that such operations enable us to construct self-dual posets from non-self-dual ones, extending the study of codes to new classes of posets.

We begin with some definitions and notation. Let $\ff_q^n$ be the $n$-dimensional linear space over $\ff_q.$ 
Let $ \cP(\n,\preceq)$ be a   poset on $\n:=\{1,2,\dots,n\}.$ 
%We will assume that $\cP$ is graded and denote by $l(i)$ the rank of $i\in\n.$ This means that $l(i)=\max \{r:i_0 \prec i_{1} \prec i_{r-1} \prec i_r=i  \}$
A subset $I\subset \n$ is called an {\em ideal} of $ \cP$ if
the relations $i\in I, j\prec i$ imply that $j\in I.$ 
An element $j\in I$ is called {\em maximal} if
there are no elements $i\in I$ such that $j\prec i.$ If $i_1,i_2,\dots$ are the maximal elements
of the ideal $I,$ we say that $I$ is generated by them, in the sense that $I$ is the smallest ideal containing $i_1,i_2,\dots,$ and write $I=\ideal{i_1,i_2,\dots}.$ Denote by
$M(I)$ the set of maximal elements of $I$.
Define the {\em dual poset} $ \cP^\bot$ on $\n$ by setting $i\prec j$ in $\cP^\bot$ whenever
$j\prec i$ in $\cP$. Ideals of $\cP^\bot$ are called {\em filters} of $\cP.$
We denote the set of all ideals and the set of all filters of $\cP$ by $\cI(\cP)$ and  $\cF(\cP),$ respectively.
Given a poset $\cP$ on $\n$, we say that a subset  $\{ i_1,i_2,\dots i_r \}\subseteq \n$ is a {\em chain of length } $r$ if $i_1\prec i_2 \prec \cdots \prec i_r$. Given $i\in\n$, we define the {\em level }  $l(i)$ of $i$ as the number of elements in any maximum-length chain that has $i$ as the maximal element. 

Given a vector $x\in\ff_q^n$ we define the support of $x$ by $\supp(x)=\{i: x_i\ne 0\}.$
Let $\ideal x$ denote the smallest-size ideal $I$ such that $\supp(x)\subseteq I.$
Call $\omega_\cP(x)=|\ideal x|$ the \emph{poset weight} of the vector $x\in \ff_q^n.$ Poset metric spaces
were introduced in \citep{bru95} where it is proved that $\omega_\cP(x)$ is indeed a weight function,
and so $d_\cP(x,y)=\omega_\cP(x-y)$ is a well-defined metric. The corresponding metric space will be denoted
by $X=(\ff_q^n,d_{\cP}).$ A {\em linear poset code} is a linear subspace of $(\ff_q^n,d_{\cP}).$ 

Poset metrics are invariant by translations, which makes them suitable for
studying linear codes, since many of the aspects of the usual theory remain valid
(the minimum distance being equal to the minimum weight, existence of syndrome
decoding schemes, etc.). On the other hand, many metric results that hold for
the Hamming space, can fail for a poset metric.
For instance, the well-known equation $\rho= \lfloor \frac{d-1}%
{2} \rfloor ,$ which relates the minimum distance $d$ of a code with its
packing radius $\rho$, not only is not valid for general posets, but those
quantities may not at all be related: there are linear codes with equal minimum
distance but different packing radii \citep{Lucas13}.

\vspace*{.1in}
{\em Examples:} 1. Let $\cP$ be an antichain on $\n$, i.e., no two elements of $\n$ are comparable. The
metric $d_\cP$ is the familiar Hamming distance of coding theory.

\vspace*{.05in}  2. Let $\cP$ be a linear order (a chain), i.e., $1\prec 2\prec\dots\prec n.$ The distance induced on $\n$ by $ \cP$ has been studied in \citep{gut98,skr01,pinheiro12}, while \citep{FPM} gave a complete classification of linear codes in $(F_q^n,\cP)$, establishing a canonical form for each class of such codes. 

\vspace*{.05in} 3. Suppose that $\n$ is a disjoint union $\n=H_0\cup H_1\cup\dots \cup H_m$ with the relation given by $i\prec j$ if and only if $i\in H_s,j\in H_t$ and $s<t.$ This order defines a
{\em hierarchical poset} $\cP$ on $\n$ which includes the above two examples as particular cases.
The metric space $(\ff_q^n,d_{\cP})$ turns out to be the only instance of poset metrics for
which the weight distribution of a linear code $\cC$ is determined by the weight
distribution of its dual code $\cC^\bot$ \citep{kim05a,kim05b,pinheiro12}.

\vspace*{.05in} 4.  Let $n=mr$ and let $\n$ be a disjoint union of $m$ chains of length $r$. This example, which also
includes the first two ones, is actually the first poset distance beyond the Hamming metric
to be studied in combinatorics, see Niederreiter \citep{nie91} and Rosenbloom and Tsfasman \citep{ros97}.  
The arising metric space is called the ordered Hamming space or the NRT space. 
It finds applications in numerical analysis \citep{nie04,skr01} and coding theory \citep{nie01b}.
Combinatorial structure of the NRT space was studied in detail in \citep{mar99,bar09b,Alves2011}.

\vspace*{.05in} 5. Let $n=1+d_0+d_0d_1+\dots+d_0d_1\dots d_{m-2}$ and let $\n$ be a level-regular rooted tree
in which every vertex $i$ with $l(i)=s$ has $d_s$ sons, $s=0,1,\dots,m-2$. Define $\cP$ by $i\prec j$ iff the vertex $i$ lies
on a path from the root to the vertex $j$. We will explore this example in detail in the next section (see Proposition \ref{prop:extension}), using a labelling of the vertices that is different from the encoding
of trees considered in earlier literature beginning with \citep{Reed72}: there, each vertex is 
labelled individually with the purpose of arranging the vertices in a linear order, while we are concerned with the order $\cP$.

Note also that encoding the nodes of a tree
is related to the problem of genetic testing for hereditary diseases:
namely, if an individual is found to be positive for a condition, this discovery supersedes the results of testing of
his ancestors (which could have missed the presence of the same condition). If the family
is represented by the ancestry tree, with each node labeled by 1/0 according as the 
individual tests positive or negative for this condition, then the metric on the family codewords is 
exactly the poset metric on a tree.

\vspace*{.05in} {\em Poset metric automorphisms.} 

A poset automorphism is a permutation $\phi:\n\to\n$ such that $x\preceq y$ if and only if 
$\phi(x)\preceq\phi(y).$
Let $S_\cP$ be the automorphism group of the poset $(\n,\cP)$.
For instance, for the Hamming space $S_\cP={\mathfrak S}_n$ (the symmetric group on $n$ elements) 
while for a single chain, $S_\cP=\{{\text {id}}\}.$ A poset $\cP$ is called {\em self-dual} if it is isomorphic to its dual $\cP^\bot,$ i.e., 
if there exists a 
permutation on $\n$ such that if $i\preceq_\cP j$ then $\pi(j)\preceq_{\cP^\bot}\!\!\pi(i).$
To give an example, let us represent $\cP$ by its {\em Hasse diagram} in which vertices correspond to
elements of $\n$ and there is an edge connecting vertices $i$ and $j$ if and only if $i\prec j,$ and  
$i\prec j' \preceq j$ implies $j' =j$.%\vspace*{-.2in}

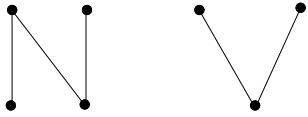
\setlength{\unitlength}{0.00040in}
\begin{center}{\begingroup\makeatletter\ifx\SetFigFont\undefined%
\gdef\SetFigFont#1#2#3#4#5{%
  \reset@font\fontsize{#1}{#2pt}%
  \fontfamily{#3}\fontseries{#4}\fontshape{#5}%
  \selectfont}%
\fi\endgroup%
{\renewcommand{\dashlinestretch}{30}
\hspace*{.75in}\begin{picture}(1141,1409)(0,-10)
\put(83,1327){\blacken\ellipse{120}{120}}
\put(83,1327){\ellipse{120}{120}}
\put(68,67){\blacken\ellipse{120}{120}}
\put(68,67){\ellipse{120}{120}}
\put(1043,82){\blacken\ellipse{120}{120}}
\put(1043,82){\ellipse{120}{120}}
\put(1073,1327){\blacken\ellipse{120}{120}}
\put(1073,1327){\ellipse{120}{120}}
\path(83,52)(83,1327)(1058,52)(1058,1327)
\end{picture}\hspace*{.5in}{\renewcommand{\dashlinestretch}{30}
\begin{picture}(1471,1439)(0,-10)
\put(803,67){\blacken\ellipse{120}{120}}
\put(803,67){\ellipse{120}{120}}
\put(68,1327){\blacken\ellipse{120}{120}}
\put(68,1327){\ellipse{120}{120}}
\put(1403,1357){\blacken\ellipse{120}{120}}
\put(1403,1357){\ellipse{120}{120}}
\path(83,1297)(803,37)(1388,1327)
\end{picture}}}}
\end{center}
Of the two posets in this figure the left is self-dual
while the right is not. The NRT poset is self-dual, including
the case of the single chain. 

\vspace*{.05in}
Since we are interested in linear codes, we concentrate on linear automorphisms of poset metric spaces.
The group of linear isometries of $X=(\ff_q^n,d_\cP),$ denoted $GL_\cP(n),$ is formed
of linear operators $T:X\to X$ such that  $d_\cP(T(x),T(y))=d_\cP(x,y)$ for all $x,y\in X.$ 
The group $GL_\cP(n)$ was characterized in \citep{pan09b}.
It can be constructed
as a semidirect product $G_{\cP}\ltimes Aut(  \cP)  $ where
$G_{\cP}$ is the set of $n\times n~$\ matrices $A=(  a_{ij})  $ over $\ff_q$ such that
 \begin{equation}\label{eq:group}
a_{ii}   \neq0, i=1,\dots, n;\;\;
a_{ij}=0\text{ if }i>j;\;\;
a_{ij}=0\text{ if }i<j\text{ and }i\notin\langle j\rangle
\text{.}%
  \end{equation}
In particular, for the Hamming space, $GL_\cP(n)= (\ff_q^\ast)^n \ltimes S_n $, for the linear order we have
$GL_\cP(n)=M_n,$ the group 
of upper-triangular matrices with nonzero main diagonal, 
and for the NRT space, $GL_\cP(n)=(M_r)^m \ltimes S_m$ (the last result is due to \citep{lee03}).

\vspace*{.1in}{\sc Nomenclature:}\\
$\n=\{1,2,\dots,n\}$\\
$\cP$ -- poset on $\n$\\
$\omega_\cP, d_{\cP}$ -- poset weight, poset distance\\
$\cI(\cP)$ the set of ideals of $\cP$, $\cF(\cP)$ the set of filters of $\cP$\\
$X=(\ff_{q}^{n},d_{\cP})$ -- poset metric space\\
$\ideal{a_1,a_2,\dots}$ -- ideal generated by $a_1,a_2,\dots\in\n$\\
$\supp(x)=\{i: x_i\ne 0; x\in \ff_q^n\}$\\
$\ideal x \triangleq\ideal{\supp(x)}$\\
$M(I)$ -- set of maximal elements of the ideal $I$\\
$\cP^\bot$ -- dual poset of $\cP$\\
$GL_{\cP}(n)$ -- group of linear isometries of $X$\\
$\cX=X/\sim$ -- the set of orbits of $GL_{\cP}(n)$\\
$G_{\cP}$ -- subgroup of $GL_{\cP}(n)$ that fixes the poset\\
$\tilde I$ -- orbit of $I$ under $Aut(\cP)$\\
$l:\cP\to\integers_+\cup\{0\}$ -- level (or rank) function for finite posets.

\section{Extension property}

Two ideals $I,J\in \cI(\cP)$ are called isomorphic, denoted $I\sim J,$ if there is
a bijection $g:I\to J$ that preserves the order. 
Given a poset $\cP=(\n,\preceq)$, an ideal $I\in \cI(\cP),$ and a poset isomorphism $\sigma\in Aut(\cP),  $
it is clear that $\sigma (  I ) \sim I$ because $\sigma|_I$ is a
poset isomorphism. However, the converse is not always true: given two isomorphic ideals 
$I,J\in\cI(\cP)$ there does not always exist an automorphism $\sigma\in Aut (  \cP )$ 
such that $\sigma(I)  =J$. The simplest
example of this situation is the poset on $ \{  1,2,3 \}  $
determined by the relation $1\prec3;$ then $ \{
1 \}  $ and $ \{  2 \}  $ are isomorphic ideals, but $Aut (
 \cP )=\{\text{id}\}.$ 

Given a poset $\cP$ and an ideal $I\in\cI(\cP)$, we denote
by $\tilde{I}$ the set of ideals that are isomorphic to $I$:
   $$
\tilde{I}= \{  J\in\mathcal{I} (  \cP )  ;I\sim J \}.
  $$
We remark that $\sim$
is an equivalence relation on $\mathcal{I} (  \cP )  $. In a similar
way, given $I\in\mathcal{F} (\cP )  $, we define
   $$
\tilde  I^\bot= \{  J\in\mathcal{F} (   \cP )  ;I\sim
J \}  \text{.}%
   $$

\begin{definition}
We say that a poset $ \cP= (     \n  ,\preceq )  $ has the
\emph{ideal-extension (IE) property }if, for every $I,J\in\mathcal{I} (
 \cP )  $, if $I$ and $J$ are isomorphic, there exists $\sigma\in Aut (
 \cP )  $ such that $\sigma (  I )  =J$. We say that $ \cP= (
 \n  ,\preceq )  $ has the \emph{filter-extension (FE)
property }if the same holds true when ideals are replaced with filters.
\end{definition}
A different and much stronger extension property (where the two sets $I,J\subseteq [n]$ need not to be  ideals) was previously studied in the infinite case when such posets are
called homogeneous \citep{Schmerl79}, leading to classification of all such posets.

We say that $(\ff_q^n ,d_{\cP})  $ has the \emph{orbits determined by ideals }(the
$\tilde {I}$-property)  if for any $x\,,y\in
\mathbb{F}_{q}^{n}$ there is $T\in GL_{ \cP} (  n )  $ such that
$T(x)=y$ if and only if $\ideal x\sim\ideal y.$
When the orbits are determined by filters, we say that $ (
\mathbb{F}_{q}^{n},d_{ \cP} )  $ has $\tilde  I^\bot$-property.

In the next proposition we show that the $\tilde {I}$-property (a property
of the vector space) and the IE-property (a property on the poset) are
essentially the same:

\begin{proposition}\label{prop:extension}
A poset $\cP$ has the $IE$-property iff $(\ff_q^n,d_{\cP})$ has the $\tilde I$-property.
A poset $\cP$ has the $FE$-property iff $(\ff_q^n,d_{\cP})$ has the $\tilde  I^\bot$-property.
\end{proposition}
\begin{proof} Assume that $\cP$ has the $IE$-property. Let $(e_{i})$ be the standard basis of $\ff_q^n.$
As shown in \citep[Theorem 1]{pan09b}, given $T\in GL_{\cP}(  n)  $, the
map $\phi_{T}:\n\to\n $ defined by
  $$
\phi_{T}(i)  =M(\langle T(e_{i})\rangle),
  $$
is a poset automorphism,   So, given $T\in GL_{\cP}(n)$ such that 
$T(x)  =y$, we have that $\phi_{T}$ is a poset automorphism and clearly
$\phi_{T}(  \ideal x )  =\ideal y,$ so that $\langle x\rangle \sim\langle
y\rangle $. This establishes the only if part.

Suppose now that $\langle x\rangle \sim\langle y\rangle$. 
Since $\cP$ satisfies the extension property, there is $\phi\in Aut(
\cP)  $ such that $\phi(\langle x\rangle )=\langle y\rangle $. Let 
$T_{\phi}:\mathbb{F}_{q}^{n}\to \mathbb{F}_{q}^{n}$ be defined by $T_{\phi}(  x_{1}
,...,x_{n})  =(  x_{\phi(  1)  },...,x_{\phi(
n)  })  $. Clearly, $T_\phi\in  GL_{\cP}(  n)  $. 

By abuse of notation, we write $M(x)$ to refer to the set of maximal elements of $\ideal x.$
Given $x\in\mathbb{F}_{q}^{n}$,
denote by $\hat x=(\hat x_1,\dots,\hat x_n)$  the vector that satisfies the following conditions: 
\emph{(i) }$M(x)  =M( \hat x)  $, 
\emph{(ii) }$\supp(\hat x)  =M(x)  $, 
\emph{(iii)} If $i\in\supp(\hat x)$ then $\hat x_{i}=1$. 

Consider a matrix $A=(a_{ij}) \in G_{\cP}$ such that $a_{ii}=x_{i}^{-1}$ if $i\in M(x)  $, 
$a_{ii}=1$ if $i\notin M(x)  $ 
and $a_{ij}=0$ for $i\neq j$. 
Let $Ax=(  x_{1}^{\prime},...,x_{n}^{\prime}),$ then
$x_{i}^{\prime}=1$ if $i\in M(  x)  $ and $x_{i}^{\prime}=x_{i}$ otherwise. 
Now consider the matrix $B=(
b_{ij})  \in G_{\cP}$ defined as follows:%
\begin{align*}
b_{ij}  &  =1 {\mbox{ if } i=j}\\ 
b_{ij}  &  = {  - x_{i}^{-1}\text{ if }x_{i}\neq 0\text{ and }%
j=\max \left\{k\in M(x); i \prec k\right\}   } \\
b_{ij}  &  =0\text{ otherwise,}%
\end{align*}
{where $\max$ refers to the usual order $\leq$ of the natural numbers.}

Let $T_x:=BA$. By construction we have that $B\in G_{\cP}$ and $BAx=\hat x.$ Concluding,
 we have that $T=T_{y}^{-1}\circ T_{\phi}\circ T_{x}$ is a linear
isometry that satisfies $T(x)  =y$. This completes the proof.
\end{proof}

\textbf{Remark.} The IE property does not necessarily imply the FE property. Let $\cP$ be a binary regular rooted tree with vertices $\{1,2,...,7\}$ labelled so that $1\prec 2, 3; 2\prec 4,5$ and $3\prec6,7$.  In the next proposition we shall prove that it has the IE property. However, $I = \{4,5\}$ and $J = \{5,6\}$ are two isomorphic filters  but there is no $\sigma \in Aut(\cP)$ such that $\sigma(I) = J.$ 

\vspace*{.05in} Posets in the five examples given above satisfy the IE property.
For hierarchical posets, linear isometries act transitively on spheres of a 
fixed radius around zero. At the same time, for hierarchical posets, ideals are isomorphic if
and only if they have the same cardinality, so that orbits of linear isometries are determined by the weight.
Therefore, the IE property is satisfied.

In contrast, the NRT posets also satisfy the IE property, but the cardinality of an ideal is not sufficient to characterize it, or equivalently, orbits of linear isometries are not determined by the weight. To do so, it is convenient to introduce
a new invariant, \emph{shapes of ideals} and observe that linear isometries act transitively on vectors
of the same shape (more on this below). 

We now show that the extension property also holds for level-regular rooted trees.

\begin{proposition}\label{prop:lt}
 Level-regular rooted tree posets possess the IE  property.
\end{proposition}
\begin{proof} 
Let $l(\cdot)$ be the rank function associated with the natural grading of $\cP.$
Let $r$ be the height of the tree, i.e., $r=\max\{l(i):i\in \n \}$ and suppose each element of rank $i<r$ has $d_i$ descendants. 
Introduce a labeling of the vertices that associates a string of integers with a vertex $a\in\n.$ Namely, if $l(a)=j\ge 1,$ then
the label $\lambda(a)=\alpha_1\alpha_2\cdots \alpha_j,$ where $\alpha_m\in \{0,1,\cdots d_m-1\}, m=1,\dots,j$. By definition $\lambda(\text{root})=\emptyset.$ The labeling is assigned 
in such a way that two elements $a,b\in \n$ with $\lambda(a)=\alpha_1\alpha_2\cdots \alpha_j$ and 
$\lambda(b)=\beta_1\beta_2\cdots \beta_k$  satisfy $a\preceq b$ iff $j\leq k$ and $\alpha_i=\beta_i$ for all $i=1,2,\dots ,j.$
In this case we can write $\lambda(b)=(\lambda(a)|\beta_{j+1}\cdots \beta_k).$  A labeling with such property is said to be \emph{consistent with the order} $\cP$. By abuse of notation, below we sometimes use labels to refer to vertices.

Let $I,J\in\cI(\cP)$ be two isomorphic ideals, and let $\phi$ be the corresponding isomorphism.
We are going to construct $\phi^\ast\in Aut(\cP)$ such that the restriction $\phi^\ast|_I=\phi.$ For $a\in I$ set $\phi^\ast(a)=\phi(a).$
Given $a\in\n\backslash I,$ consider the chain from the root to $a$. This chain is unique and intersects $I$ because $\cP$ 
is a rooted tree. Let $a_I$ be the last vertex in this chain that is in $I$ (the ``meet" of $a$ and $I$). Thus we have $a_I=b_0\prec b_1\prec\dots\prec b_{l(a)-l(a_I)}=a$ for some vertices $b_1,\dots,b_{l(a)-l(a_I)-1}.$

By our construction, the label of each $b_l, l\ge 1$ is obtained by concatenating the label
$\lambda(a_I)$ with a tail formed of $l$ letters $\beta_1,\dots,\beta_{l},$ where 
$\beta_j\in \{0,1,\dots,d_{l(a_I)+j}-1\},j=1,\dots,l.$
For $a\in I$ define the set of descendents of $a$ not contained in $I$: 
  $$
\Lambda_{a,I}= \{0 \leq j \leq d_{l(a)}-1: \; (\lambda(a)|j) \notin I\}
 $$
(this set can be empty).
Since $I$ and $J$ are isomorphic, and since $\cP$ is level-regular, we have
\[  
| \Lambda _{a,I} | = | \Lambda _{\phi (a),J}| , \quad a\in I,
\]
so for each $a\in I$ there is a bijection $\gamma_a:   \{0,1,\dots,d_{l(a)}-1\} \rightarrow 
\{0,1,\dots,d_{l(\phi (a))}-1\}$ such that $\gamma_a \left( \Lambda _{a,I} \right) = \Lambda _{\phi (a),J}$ and $(\lambda (\phi (a)) |\gamma_a(j)) = \lambda (\phi (\lambda (a)|j)) $ for $j\in\left\{ 0,1,\dots, d_{l(a)} - 1 \right\} \backslash \Lambda_{a,I}$. In other words, $\gamma_a$ induces the same map as $\phi$ when restricted to the immediate descendants of $a$ in $I$. 
%and of $\phi (a)$ in $J$, restricted to the level $l(a)+1$.  

Now we are able to define the isomorphism $\phi^\ast$. Given $a\in \n$, consider its label 
        $$
       \lambda (a)=(\lambda (a_I)| \beta_{l(a_I)+1},\cdots ,\beta_{l(a)})
     $$ 
and define $\phi^{\ast }(a)$ to be the vertex labeled as 
    $$ 
\lambda (\phi^{\ast }(a)) = (\lambda (\phi (a_I))|\gamma_a(\beta_{l(a_I)+1}), \beta_{l(a_I)+2}, \dots \beta_{l(a)} ).
    $$
Since $\gamma_a$ is a bijection, and since $\cP$ is level-regular, $\phi^\ast$ is well defined. It is a bijection that
preserves the order because the labeling is consistent with $\cP.$ Therefore, it is an order isomorphism which also satisfies 
$\phi^\ast|_I=\phi.$ Thus the proof is complete.
\remove{Considering that \textit{(i)}  $\lambda_a$ is a bijection between $\Lambda _{a,I}$ and $\Lambda _{\phi (a),J}$ for each $a\in I$; \textit{(ii)} $\cP$ is level--regular and \textit{(iii)} the labeling $\lambda $ is consistent with the order $\cP$, we find that $\phi^{\ast }$ is an order isomorphism that extends $\phi$. Indeed, $\phi^{\ast }$ is well defined by statements \textit{(i)} and \textit{(ii)}, it is a bijection by statement \textit{(ii)} and it preserves the order by \textit{(iii)}. By construction we have that ${\phi^\ast|}_I=\phi$. 
}
\end{proof}

\subsection{Remarks on lattices}

A poset $\cP=\left(\n,\prec\right)$ is called a 
\emph{meet semilattice} if for any $x,y\in n$ there is a unique greatest lower bound $z$ of $x$ and $y$.  
We write $z=x\wedge y$ and call it the meet of $x$ and $y$. 
Semilattices admit a natural grading, and we denote by $X_i, i=0,1,\dots ,m$ its fibers, i.e., the sets of points of $\cP$ of the same rank.
A semilattice $ \cP=\left(  X,\prec \right)  $ is  said to be \emph{regular} if the following conditions
are satisfied:

\begin{enumerate}
\item Given $y\in X_{m},z\in X_{r}$ with $z\preceq y$ the number of points
$u\in X_{s}$ such that $z\preceq u\preceq y$ is a constant $\mu\left(
r,s\right)  $;

\item Given $u\in X_{s}$, the number of points $z\in X_{r}$ such that
$z\preceq u$ is a constant $\nu\left(  r,s\right)  $;

\item Given $a\in X_{r},y\in X_{m}$ with $a\wedge y\in X_{j}$, the number of
pairs $\left(  b,z\right)  \in X_{s}\times X_{m}$ such that $b\preceq
z,a\preceq z$ is a constant $\pi\left(  j,r,s\right)  $.
\end{enumerate}
Regular semilattices were introduced by Delsarte \citep{del76}; see also \citep{cec08}, Ch.8.
It is straightforward to show that level-regular rooted trees are regular semilattices and the family of semilattices 
seems to be a fertile ground for posets satisfying the IE property. 
The following example shows that not every regular semilattice satisfies the IE property.

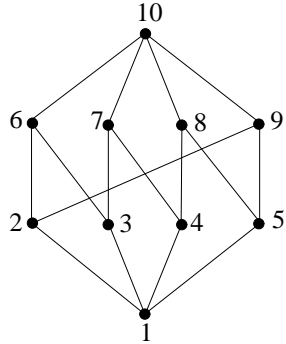
\remove{\textbf{Counterexample:}

Consider the poset $\cP$ over  $\llbracket 10 \rrbracket$ determined by the relations:
\begin{itemize}
\item$0\prec i$ for $i=1,2,3,4$ ;
\item $1\prec 8, 1\prec 5$ and  $i\prec i+3,i+4$ for $i=2,3,4$ ;
\item $i\prec 9$ for $i=5,6,7,8$.
\end{itemize}
}

\vspace*{.05in}\setlength{\unitlength}{0.00045in}
\begin{center}
\begingroup\makeatletter\ifx\SetFigFont\undefined%
\gdef\SetFigFont#1#2#3#4#5{%
  \reset@font\fontsize{#1}{#2pt}%
  \fontfamily{#3}\fontseries{#4}\fontshape{#5}%
  \selectfont}%
\fi\endgroup%
{\renewcommand{\dashlinestretch}{30}
\begin{picture}(3263,4000)(0,-10)
\put(285,2595){\circle*{120}}
\put(1185,2575){\circle*{120}}
\put(2055,2575){\circle*{120}}
\put(2965,2585){\circle*{120}}
\put(285,1415){\circle*{120}}
\put(1185,1395){\circle*{120}}
\put(2055,1395){\circle*{120}}
\put(2965,1405){\circle*{120}}
\put(1615,335){\circle*{120}}
\put(1625,3655){\circle*{120}}
\drawline(1615,335)(275,1415)(275,2635)(1615,3655)
\drawline(1625,315)(1185,1415)(265,2615)
\drawline(1625,345)(2055,1405)(1185,2585)(1635,3685)
\drawline(1615,335)(2985,1415)(2045,2595)(1625,3645)
\drawline(295,1415)(2985,2595)(1625,3655)
\drawline(1185,1415)(1185,2585)
\drawline(2045,1405)(2055,2515)
\drawline(2975,1405)(2975,2585)
\put(15,2535){\makebox(0,0)[lb]{\smash{{\SetFigFont{10}{16.8}{\rmdefault}{\mddefault}{\updefault}6}}}}
\put(975,2525){\makebox(0,0)[lb]{\smash{{\SetFigFont{10}{16.8}{\rmdefault}{\mddefault}{\updefault}7}}}}
\put(1515,3805){\makebox(0,0)[lb]{\smash{{\SetFigFont{10}{16.8}{\rmdefault}{\mddefault}{\updefault}10}}}}
\put(2195,2505){\makebox(0,0)[lb]{\smash{{\SetFigFont{10}{16.8}{\rmdefault}{\mddefault}{\updefault}8}}}}
\put(3105,2525){\makebox(0,0)[lb]{\smash{{\SetFigFont{10}{16.8}{\rmdefault}{\mddefault}{\updefault}9}}}}
\put(15,1325){\makebox(0,0)[lb]{\smash{{\SetFigFont{10}{16.8}{\rmdefault}{\mddefault}{\updefault}2}}}}
\put(1315,1305){\makebox(0,0)[lb]{\smash{{\SetFigFont{10}{16.8}{\rmdefault}{\mddefault}{\updefault}3}}}}
\put(2155,1325){\makebox(0,0)[lb]{\smash{{\SetFigFont{10}{16.8}{\rmdefault}{\mddefault}{\updefault}4}}}}
\put(3125,1355){\makebox(0,0)[lb]{\smash{{\SetFigFont{10}{16.8}{\rmdefault}{\mddefault}{\updefault}5}}}}
\put(1555,15){\makebox(0,0)[lb]{\smash{{\SetFigFont{10}{16.8}{\rmdefault}{\mddefault}{\updefault}1}}}}
\end{picture}
}
\end{center}

Direct verification shows that $\cP$ is a regular meet semilattice. However, $I=\left\{1,2,3\right\}$ and $J=\left\{1,2,4\right\}$ 
are isomorphic ideals but there is no isomorphism of $\cP$ that maps $I$ into $J$, since $2$ and $3$ are covered by $6,$ but no element in level $2$ covers $2$ and $4$. This poset is also self-dual, and therefore forms a lattice.

For $a,b\in X$ define the set
  $$
   a \vee b=\{x\in X: a\preceq x, b\preceq x \text{ and }
   ( a\prec y\preceq x) \Rightarrow (y=x); \; (b\prec y\preceq x) \Rightarrow (y=x)\}.
  $$
We say that the semilattice $\cP$ is \emph{strongly regular} if it is regular and satisfies the following additional conditions:
\begin{enumerate}
\item[(r1)] For $z \in X_r, \left|\left\{u\in X_s;z\preceq u\right\}\right|$ is a constant $\bar{\nu}(r,s)$;
\item[(r2)] Given a pair of vertices $a,b\in X_r$, if $a\vee b\ne\emptyset,$ then the
quantity $\left|\left\{u\in X_s; a\vee b \preceq u\right\}\right|$ is a constant $\rho (r,s)$ 
that does not depend on $a$ and $b.$
\end{enumerate}

It is not difficult to see that a level-regular rooted tree is a strongly regular semilattice, while the 
lattice in the previous example is not (indeed, let $r=1$ and $s=2,$ then the pairs $(2,3)$ and $(2,4)$ give a counterexample
to condition (r2)). We conjecture that strongly regular semilattices satisfy the IE property.

\section{Association schemes on poset metric spaces}
\subsection{Association schemes on poset metric spaces} MacWilliams-type relations between weight enumerators of additive codes
and their dual codes can be derived using Delsarte's theory of association schemes.
We briefly summarize the facts about association schemes used below, following the presentation in \citep[ch.~2]{bro89}.
%In the form most useful for our needs this theory is presented in Brouwer et al. \citep{bro89},
%which is the main source for the material this section.
Given a finite space $X$, a symmetric association scheme $\cA(X,\cR)=(R_0,R_1,\dots,R_s)$
is a partition of the set $X\times X$ into $s+1$ classes such that

%\vspace{.05in}
\nd  (i)\; $R_0=\{(x,x),x\in X\}$,

\nd (ii)\; if $(x,y)\in R_\alpha$ then $(y,x)\in R_\alpha$ for all $x,y\in X$, $\alpha=0,1,\dots,s.$

\nd\hangindent.27in \hangafter=1
 (iii)\; there are numbers $p_{\alpha\beta}^\gamma$ such that for any $(x,y)\in R_\gamma$ the number 
  of $z\in X$ with $(x,z)\in R_\alpha$ and $(y,z)\in R_\beta$ equals $p_{\alpha\beta}^\gamma,$
$\alpha,\beta,\gamma\in\{0,1,\dots,s\}.$

\vspace{.05in}\nd
The adjacency matrix $A_\alpha$ of the class  $R_\alpha$ is defined by
   $$
  (A_\alpha)_{xy}={\bf 1}_{(x,y)\in R_\alpha}, \quad \alpha=0,1,\dots,s
   $$  
	meaning that an entry $(x,y)\in A_\alpha$ is $1$ if $(x,y)\in R_\alpha$ and $0$ otherwise.
	
The matrices $A_\alpha$ generate an $(s+1)$-dimensional commutative $\complexes$-algebra 
called the adjacency algebra of $\cA.$  This algebra has a basis of primitive idempotents $(E_\alpha,\alpha=0,1,\dots,s)$.
Define matrices $P$ and $Q$, called the eigenvalues of $\cA$, by
   \begin{equation}\label{eq:ev}
   A_\beta=\sum_{\alpha=0}^s P_{\alpha\beta}E_\alpha  \quad\text{and}\quad E_\beta=\frac1{|X|}\sum_{\alpha=0}^s 
   Q_{\alpha\beta}A_\alpha.
   \end{equation}
The numbers $v_\alpha=p_{\alpha\alpha}^0$ and $m_\alpha=\rank E_\alpha$ are called the valencies and the multiplicities
of the scheme $\cA.$ 

Suppose that $X$ has the structure of an abelian group. A scheme $\cA$ is called a translation association scheme if
for all $R\in \cR$
  $$
   (x,y)\in R \;\Rightarrow\; (x+z,y+z)\in R, \quad z\in X.
  $$
In this paper we restrict our attention to the case $X=\ff_q^n,$ an $n$-dimensional linear space.

The Hamming association scheme is defined by the relations $R_\alpha=\{(x,y)\in \ff_q^n: d_H(x,y)=\alpha\}.$
If $\cP$ is a single chain (Example 2) or an hierarchical poset (Example 3), we can again define $R_\alpha=\{(x,y)\in \ff_q^n: d_\cP(x,y)=\alpha\}, \alpha=0,1,\dots,n$
and prove that these relations satisfy (i)-(iii). However, for the NRT poset (Example 4) this approach fails
to produce an association scheme. To define the NRT association scheme, also called the {\em ordered Hamming scheme}, 
let 
   \begin{equation}\label{eq:shape}
     \shape(I)=(e_1,\dots,e_r), \text{ where } e_j=|\{i\in M(I): \;l(i)=j\}|, 
     j=1,\dots,r.
  \end{equation}
  Define $\cR=\{R_e, e=(e_1,\dots, e_r)\},$ where $e$ ranges over all the $r$-tuples of 
  nonnegative integers such that $\sum_i e_i\le m.$
The relations of the scheme are given by $R_e=\{(x,y)\in (\ff_q^n)^2: \shape(x-y)=e\}.$
 This scheme was introduced in \citep{mar99} and further studied in \citep{bar09b}.
 We note that the group of linear isometries of the NRT space acts transitively on the sets
 $N_e=\{x\in \ff_q^n:\shape(x)=e\}.$

\vspace*{.1in}
Let us consider a general poset metric space $X=(\ff_q^n,d_{\cP}).$
There is more than one way to define an association scheme on $X$\footnote{This is similar to
classical coding theory: the action of the full group of linear isometries of the Hamming space 
defines the scheme
relative to the Hamming weight, while the permutation part of the group gives rise to
the scheme relative to complete weight enumerators, viz. \citep{mac91}, Sect.5.6.}. In the most
general case, to every ideal $I$ there corresponds a relation $R_I:=\{(x,y)\in X^2: \ideal{x-y}=I\}$, although this definition
is too general to be useful.
We rely on the definition that is the most relevant to the theory of linear codes.
Let $d_\cP(\cdot,\cdot)$ be a poset metric on $X$ and let $GL_\cP(n)$ be the group of linear isometries. The action of $GL_\cP(n)$ 
defines an equivalence relation on $X\times X$ where $x\sim y$ if the vectors $x$ and $y$ belong
to the same orbit, i.e. there is $T\in GL_\cP(n)$ such that $T(x)=y.$ 
Let $\cX:=X/\!\!\sim$ be the set of orbits and suppose that $|\cX|=s+1$ for some $s.$
Consider a partition $\cR=\{R_\alpha|\alpha\in \cX\}$ of $X\times X$ given by 
  \begin{equation}\label{eq:relation}
  R_\alpha=\{(x,y)\in X\times X| x-y\in \alpha\}, \quad\alpha\in \cX.
  \end{equation}

The following simple observation sets the stage for the study of linear poset codes.
\begin{proposition} \label{prop31} {\sl The pair $(X,\cR)$ forms a translation association scheme $\cA$ with $s$
classes.}
\end{proposition}
\begin{proof} We need to check the definition of the association scheme. The subsets $R_\alpha,
\alpha\in \cX$ form a partition because $\sim$ is an equivalence relation on $X$ and hence 
a partition of $X\times X.$ 

Property (i) follows because
every $T\in G$ is invertible and linear, so $T(x)=0$ iff $x=0.$

Property (ii) (symmetry) follows because $G$ is a group, and so $T\in G$ iff $T^{-1}\in G.$

Let us verify the intersection property (iii). Let $(a_1,a_2)$ and $(b_1,b_2)$
be representatives of the orbits $\alpha$ and $\beta,$  respectively.
Given $(x,y)\in R_\gamma$ denote
   $$
  p_{\alpha\beta}^{(x,y)}=\{z\in X|(x,z)\in R_\alpha \text{ and } (y,z)\in R_\beta\}.
   $$
Consider another pair $(x',y')\in R_\gamma.$ There exists $T\in G$ such that
$T(x-y)=x'-y'.$ Let $z\in p_{\alpha\beta}^{(x,y)},$ then $(x,z)\in R_\alpha$ and
$(y,z)\in R_\beta.$ We shall prove that $(T(x),T(z))\in R_\alpha$ and
$(T(y),T(z))\in R_\beta.$ By definition of $R_\alpha$ there is $S\in G$ such that
$S(x-z)=a_1-a_2.$ Then
  \begin{align*}
   T(x)-T(z)=T(x-z)=TS^{-1}(a_1-a_2)
  \end{align*}
so $(T(x),T(z))\in R_\alpha.$ Similarly, $(T(y),T(z))\in R_\beta.$ It follows that
there is an injective map $T:p_{\alpha\beta}^{(x,y)}\to p_{\alpha\beta}^{(x',y')}.$
Similarly, $T^{-1}$ defines an injective map in the reverse direction, and so
$|p_{\alpha\beta}^{(x,y)}|=|p_{\alpha\beta}^{(x',y')}|.$ This completes the proof.
\end{proof}

{We defined each relation $R_\alpha$ to be an orbit of a vector under the group of linear isometries relative to a poset metric. 
It is possible to define an association scheme on $(\ff_q^n,d_{\cP})$ in another way as follows. 
Consider an equivalence relation on $\cI(\cP)$ and denote by $\tilde I$ the equivalence class of ideals that 
contains $I$. Let $\cR=(R_{\tilde I}) $ be a set of relations on $\ff_q^n$, 
where $R_{\tilde I}=\{(x,y)\in (\ff_q^n)^2: \ideal{x-y}\in\tilde I\}.$ It is possible to prove, in a manner similar to Proposition~\ref{prop31}, that for any equivalence relation on  $\cI(\cP)$, this defines an association scheme on $(\ff_q^n,d_{\cP})$.

If the extension property holds true, and the equivalence relation on $\cI(\cP)$ is given by $I\sim J$ \emph{iff} $I$ and 
$J$ are isomorphic as posets, this definition gives the same association scheme as the one defined in Proposition 
\ref{prop31} above. } This kind of approach, considering relations between equivalent classes of ideals and orbits of vectors in  $\ff_q^n$ under the action of groups, was introduced in \citep{hyu06}; see also \citep{choi12}.

\section{Duality of schemes}

A translation association scheme $\cA$ has a dual scheme $\cA^\ast$ defined by the characters 
of the group $X.$  Characters form a multiplicative group $X^\ast$ with
the operation given by $(\chi_1\chi_2)(x)=\chi_1(x)\chi_2(x).$ It is well known that $X$ and
$X^\ast$ are canonically isomorphic as groups. 
\remove{
Consider the natural action of ${\mathfrak A}$
on $X^\ast$ induced by the action of the incidence matrices:
   $$
     A_\alpha\chi=\Big(\sum_{x\in\alpha}\chi(x)\Big)\chi
   $$
where $\alpha$ is the orbit of $x$ under the action of $G$ and $\chi=(\chi(x),x\in X)$
is a vector.} Let $\chi,\psi\in X^\ast$.
Define the relations of the dual scheme by putting $(\chi,\psi)\in R_i^\ast$ iff
$E_\alpha\eta=\eta,$ where $\eta=\chi^{-1}\psi.$
The dual translation scheme $\cA^\ast(X^\ast,R^\ast)$ satisfies the following properties \citep[p.69-70]{bro89}:

\vspace*{.05in}\nd \hangindent.3in \hangafter=1
(D1) Let $v_\alpha, m_\alpha, \alpha\in\cX$ be the valencies and multiplicities of the scheme $\cA$ and
let $P,Q$ be its eigenvalues. The scheme $\cA^\ast$ is a translation association scheme 
with $s$ classes, valencies $v_\alpha^\ast=m_\alpha,$ multiplicities $m_\alpha^\ast=v_\alpha,$ 
and eigenvalues $P^\ast=Q,Q^\ast=P.$

\vspace*{.05in}\nd \hangindent.3in \hangafter=1
(D2) Let $N_\alpha=\{x\in X| (x,0)\in R_\alpha\}, N^\ast_\alpha=\{\chi\in X^\ast| E_\alpha\chi=\chi\}.$ Then 
$v_\alpha=|N_\alpha|, m_\alpha=|N_\alpha^\ast|,$
    \begin{equation}\label{eq:PQ}
  P_{\alpha\beta}=\sum_{x\in N_\beta} \chi(x), \;\chi\in N_\alpha^\ast, \quad 
  Q_{\alpha\beta}=\sum_{\chi\in N_\beta^\ast} \chi(x),\; x\in N_\alpha.
   \end{equation}
 
%Moreover, for each character $\chi$ of $\ff_q^n$ (considered as a vector indexed by $\ff_q^n$) we have \citep[p.69]{bro89}
%  $    A_j\chi=\big( \sum_{x\in N_j}\chi(x)\big)\chi
%  $
%Multiplying this on the right by $\chi^\ast,$ suming on $\chi^\ast\in N_\alpha^\ast$ and using \eqref{eq:ev}, we obtain
%the well-known expression
\vspace*{.05in}\nd \hangindent.3in \hangafter=1
(D3) $E_\alpha=\frac1{|X|}\sum_{\chi\in N_\alpha^\ast}\chi \chi^\dag, \;\alpha=0,1,\dots, s.$

\vspace*{.05in}
Two $s$-class association schemes $\cA(X,\cR)$ and $\cB(X',\cR')$
are called {\em isomorphic}
if there is a bijection $\phi: X\to X'$ such that $(x,y)\in R_\alpha$
if and only if $(\phi(x),\phi(y))\in R'_{\pi(\alpha)}$ for some fixed permutation of the indices
in $\cR$ and $\cR'.$ If $\cA$ and $\cB$ are translation schemes, then $\phi$ agrees with
the translations.
%The scheme $\cA$ and its dual may or may not be isomorphic; if they are, then $\cA$ is called self-dual. 
%The inner distribution of a linear code is the set $\{b_i,i=0,1,\dots,s+1

\subsection{Self-dual posets}
Our motivation to study self-dual association schemes on posets comes from duality
of linear codes. 
Let $\cC\subset X$ be a linear code. The dual code of $\cC$ is the subgroup $\cC^\ast=\{\chi\in X^\ast|
\chi(x)=1$ for all $x\in \cC\}$. Even though $X\cong X^\ast,$
the codes $\cC$ and $\cC'$ live in different association schemes that are not necessarily
isomorphic. In classical coding theory problems, the dual code $\cC^\bot$ is defined with 
respect to an inner product on $X\times X,$ whereupon the codes $\cC^\bot$ and $\cC^\ast$ 
are identified with the help of the isomorphism of the dual groups. 
This is consistent for the Hamming scheme, but generally such identification
does not necessarily preserve the scheme structure.

\vspace*{.05in}
{\em Example 2 (continued):} Let $\cR=\{R_0,R_1,\dots,R_n\},$ where $R_i=\{(x,y): l(M(x-y))=i\}$ for all $i$.
The dual scheme can be realized by relations on characters that are defined by the dual chain $\cP^\bot=(1\succ 2\succ\dots\succ n)$
\citep{mar99}.

\vspace*{.05in}
Given a poset metric space $X,$ define the scheme $\cA^\bot=\cA(X,\cR_{\cP^\bot})$
with respect to the action of the group $GL_{\cP^\bot}(n).$ 
Duality of linear codes will be consistent with this definition
if $\cA^\ast\cong\cA^\bot$. Natural candidates for this to hold are self-dual posets,
in which case, of course, $\cA\cong\cA^\bot.$ 

This discussion motivates the following theorem.
\vspace*{.05in}\begin{theorem}\label{thm:duality} 
{\sl Suppose that $\cA$ is a translation association scheme on $X$ whose classes
are given by orbits of the group $GL_{\cP}(n)$ of linear isometries of a poset metric space $(X,\cP).$
Then $\cA^\ast\cong\cA^\bot$ if and only if $\cP$ is self-dual.}
\end{theorem}

\vspace*{.05in}{\em Proof:} 
%We have 
%  $  \chi_x(y)=\omega^{(x,y)},   $
%where $\omega$ is a primitive $q$th degree root of unity and $(x,y)=\sum_{i=1}^n x_i y_i.$
%\cite[p.~71-72]{}.
The ``if'' part follows straightforwardly from self-duality of $\cP$.
Formally, let $\alpha\in \cX$ be an orbit of $GL_\cP(n).$ 
For $x\in X$ denote by $\chi_x\in X^\ast$ its image under the isomorphism $X\cong X^\ast.$
%will be established if we find a bijective map $\phi:X\to X^\ast$
%such that $(x,y)\in R_\alpha$ if and only if $(\chi_{\phi(x)},\chi_{\phi(y)})\in R^{\ast}_{\pi(\alpha)},$
%for some fixed permutation $\pi$ of the orbits. 
Let $\tau$ be a permutation on $\n$ that
maps $\cP$ to $\cP^\bot$ and let $T_\tau$ be the corresponding $n\times n$ permutation matrix. 
The matrix $\phi=T_\tau$ defines a linear isometry on $X.$ 
\remove{Put $\phi=T_\tau$ and
consider the sets $\tilde R_\alpha^\ast=\{(\chi_{\phi(x)},\chi_{\phi(y)})| (x,y)\in R_\alpha\},
\tilde N_\alpha^\ast=\{\chi_{\phi(x)}|x\in N_\alpha\}, \alpha\in\cX.$ }

We have $(\chi_x,\chi_y)\in R_\alpha^\ast$ iff $\eta:=\chi_x^{-1}\chi_y$ is in eigenspace $\alpha,$
i.e., $E_\alpha \eta=\eta,$ or, using (D3)
  $$
  \frac1{|X|}\sum_{\chi\in N_\alpha^\ast}\chi \chi^\dag\eta={\bf 1}(\eta\in N_\alpha^\ast),
  $$
i.e., iff $(x-y)\in\alpha.$
Thus for the adjacency matrices of $\cA^\ast$ we have $(A_\alpha^\ast)_{xy}={\bf 1}((x-y)\in\alpha).$
At the same time, if $A^\bot$ is realized as the scheme on $\cI(\cP^\bot),$ then the orbit $\alpha$
is transformed into an isomorphic orbit $\phi(\alpha)$ with respect to the action of $GL_\cP(n)$ on $\cP^\bot.$
Thus, we have $(A^\bot_{\phi(\alpha)})_{xy}={\bf1}((x-y)\in\alpha),$ i.e., $\cA^\ast{\cong_\phi} \cA^\bot.$

\remove{
Consider the matrices
  $$
    \tilde E_\alpha=\frac 1{|X|}\sum_{\chi\in \tilde N_\alpha^\ast} \chi\chi^\dag.
  $$
These matrices are symmetric positive semidefinite, and have the following properties
for any $\alpha\in \cX:$
  $$
     (a)\,\tilde E_\alpha^2=\tilde E_\alpha, \; 
    (b)\,\tilde E_\alpha\eta=\eta \text{ iff } \eta\in  \tilde N_{\alpha}^\ast,\;
      (c)\,A_\alpha =\sum_{\beta\in \cX} P_{\beta\alpha}\tilde E_\beta
  $$
These properties are verified by straightforward calculations. Property (a) is immediate.
To prove (b), let $y\in X$ and
compute
   \begin{align*}
  (\tilde E_\alpha\eta)(y)&=\Big(\frac 1{|X|}\sum_{\chi\in \tilde N_\alpha^\ast} \chi\chi^\dag\eta\Big)(y)
=\frac 1{|X|}\sum_{\chi\in \tilde N_\alpha^\ast} \sum_{z\in X}
\chi(y)\chi^\dag(z)\eta(z) \\
&=\frac1{|X|}\sum_{x\in N_\alpha}\chi_{\phi(x)}(y)\sum_z\chi_{\phi(x)}^\dag(z)\eta(z)
=1(\eta\in \tilde N_\alpha^\ast)\eta(y),
  \end{align*}
where the last step holds because the sum on $z$ equals $|X|$ or 0 according as $\chi_{\phi(x)}=\eta$
or not. By self-duality of $\cP$, these matrices coincide with the primitive idempotents of 
 we conclude that 
these matrices form the set of primitive idempotents of $\cA,$ i.e., 
$\tilde E_\alpha=E_\alpha,$ and therefore, $\tilde R_\alpha^\ast=R_\alpha^\ast.$}

The ``only if'' part will follow from Proposition \ref{prop3} below. Let $\cQ$ be a poset on $\n$ and let $T$ be a linear isometry
of the poset metric space $X=(\ff_q^n,\cQ).$ Consider the poset metric space on $X_T=({\mathbb{F}}_{q}^{n},\cQ_T)$ 
induced by $\cQ$ and $T$. Namely, given a vector $x=\sum_{i}x_{i}T(e_{i})$ we define the weight $|x|_{\cQ,T}:=|\ideal{x}_{\cQ}|.$
%define the weight $\omega_{Q,T}(x)=|\langle\{i|x_{i}\neq 0\}\rangle_{Q}|$
%where $\langle.\rangle_{Q}$ refers to the smallest $Q$-ideal that contains the
%subset argument. Define the distance $d_{Q,T}(x,y)=\omega_{Q,T}(x-y).$

\begin{lemma}
\label{lem}The poset metric spaces $X$ and $X_T$ are isometric
and thus the association schemes $\cA_\cQ$ and $\cA_{\cQ,T}$ are isomorphic.
\end{lemma}
\begin{proof}
Indeed, given $x=\sum_{i}x_ie_{i}\in\mathbb{F}_{q}^{n}$
we have 
\begin{align*}
\omega_{\cQ,T}(T(  x)  )&  =\omega_{\cQ,T}\big(  T(
\sum_{j=j}^{n}\alpha_{j}e_{j})  \big) 
=\omega_{\cQ,T}\big(  \sum_{j=1}^{n}\alpha_{j}T(  e_{j})  \big)  =\vert \langle \{  j|\alpha_{j}\neq0\}  \rangle
_{\cQ}\vert \\
&  =\omega_{\cQ}(  x). \hfill
\end{align*}
The lemma is proved.
\end{proof}

\begin{proposition}
\label{prop3}Let $\cP$ and $\cQ$ be two posets defined on $\n$. Consider the
metric spaces $X=(\ff_q^n,d_{\cP})$ and $X'=(\ff_q^n,d_{\cQ)}$ and let $GL_\cP(n)$ and $GL_\cQ(n)$ be their
groups of linear isometries. Suppose that the translation association schemes
$\cA_\cP$ and $\cA_{\cQ}$ defined by these groups are isomorphic and the isomorphism $\phi:X\to X^{\prime}$ is
linear. Then the poset metric spaces $(\ff_q^n,d_{\cP})$ and $(\ff_q^n,d_{\cQ})$ are isometric
and the posets $\cP$ and $\cQ$ are isomorphic.
\end{proposition}

\begin{proof}
The proof is by induction on $n$. The base case $n=1$ is straightforward.
Suppose that the statement is true for every pair of posets $\cP^{\prime}$ and
$\cQ^{\prime}$ defined on $\fs{n-1}.$ Let $\cP$ and $\cQ$ be posets defined on $\n$
and suppose that $\mathcal{A}_{\cP}\cong\mathcal{A}_{\cQ}.$ In other words, there
is a linear bijection $\phi:X\to X$ such that $(x,y)\in R_{\cP,\alpha}$ if and only
if $(\phi(x),\phi(y))\in R_{\cQ,\pi(\alpha)}$ for some bijection $\pi$ between the
indices in $\cR_{\cP}$ and $\cR_\cQ.$ Consider a basis
$\beta=\{e_{1},\dots,e_{n}\}$ of $X$ such that $\supp(e_{i})=\{i\}.$
Since $\phi$ is linear, it preserves linear independence, and it follows that
$\{\phi(e_{1}),\dots,\phi(e_{n})\}$ is linearly independent and hence a basis
of $X.$

Denote by $\delta_{\cP}(i,j)$ the graphical distance between $i$ and $j$ in the
Hasse diagram of $\cP$ and let 
   $$
   J_{\cP}^{-}(i):=\{j\in\n\mid \delta_\cP(i,j)=1,j\prec i\}.
   $$
   Let $i$ by a maximal element in $\cP$ and suppose w.l.o.g.
that $i=n.$ This can always be accomplished by isomorphically relabeling $\n.$
Indeed, let us denote by $\theta$ such an isomorphism and let
$T=T_{\theta}:\mathbb{F}_{q}^{n}\to\mathbb{F}_{q}^{n}$ be the linear
map induced by $\theta$: $T(\sum x_{i}e_{i})  =\sum x_{i}e_{\theta(i)}.$  
We note that the metric on the new poset $\cP'$ is induced by $\cP$ and $T$. Lemma \ref{lem} implies
that $(X,\cP)$ and $(X,\cP_T)$ are isometric and give rise to isomorphic association schemes, so our assumption
that $n$ is maximal is justified.
In a similar manner, we may assume that (again w.l.o.g.) that $j=n-1\in J_{\cP}^{-}(n),$ or, 
in other words, that $n-1\in\langle n\rangle _{\cP}$.

We claim that there is a $\cP$-isometry that takes $e_{n}=(e_{n-1}+e_{n})-e_{n-1}$ to
$e_{n}-e_{n-1}$. Indeed, define an $n\times n$ matrix $M=(a_{ij})$ by
follows:%
\begin{align*}
a_{ii}  &  =1,\quad \forall i\in\n \\
a_{n-1,n}  &  =-1\\
a_{ij}  &  =0 \quad\text{ otherwise.}%
\end{align*}
It is clear that $A( e_{n}) =e_{n}-e_{n-1}$. Moreover, since
$n-1\prec n$ we have that $A\in GL_\cP(n)$ and is an isometry. 
Since $e_n=(e_{n-1}+e_{n})-e_{n-1},$ we conclude that the pairs 
$(e_{n-1},e_{n})$ and $(e_{n-1},e_{n-1}+e_{n})$ belong  to the same relation
$R_{\cP,[e_{n}]}\in\mathcal{R}_{\cP}.$

Letting $f_{i}=\phi(e_{i})$ for all $i$, we conclude that the pairs
$(f_{n-1},f_{n})$ and $(f_{n-1},f_{n}+f_{n-1})$ belong to the same class
$R_{\cQ,[f_{n}]}\in\cR_\cQ.$ Therefore, there is an isometry that sends
$f_{n}=(  f_{n}+f_{n-1})  -f_{n-1}$ to $f_{n}-f_{n-1}$. Indeed,
given a matrix $B=(  b_{ij})  \in G_{\cQ}$, we have that $B(f_{n})$ 
has coordinates $(  b_{1n},b_{2n},...,b_{n-1,n},b_{nn})$ with respect to the
basis $\{  f_{1} ,...,f_{n}\}  $ . Since the coordinates of $f_{n}-f_{n-1}$ are $(
0,...,0,1,1)  $ and since $B\in G_{\cQ}$ get that $n-1\prec_{\cQ}n$, or in
other words, $n-1\in J_{\cQ}^{-}(n)$. The same reasoning can be applied to
every $i\in J_{\cP}^{-}(  n)  $, so that we find that $J_{\cP}^{-}(  n)  $ is mapped into $J_{Q}^{-}(  n)  $. 
Using the inverse isomorphism, we can conclude that actually this map between
$J_{P}^{-}(  n)  $ and $J_{Q}^{-}(  n)$ is bijective.

The induction hypothesis ensures that the posets $(\fs{n-1},\cP)$ 
and $(\fs{n-1},\cQ)$ are isomorphic.
Since $J_{\cP}^{-}(  n)  =J_{\cQ}^{-}(n)  $, we conclude that $(\n  ,\cP)\cong(\n,\cQ)$.
This concludes the proof.
\end{proof}

\subsection{Remarks on the parameters} Let $\cA=(X,\cR)$ be a translation association
scheme whose relations are indexed by the orbits of the isometry group of a poset metric
space $X=(\ff_q^n,d_\cP)$. The parameters of $\cA$ can be found from Properties (D1)-(D3). 
Let $GL_{\cP}(n)$ be the isometry group of $X$ whose structure is given by \eqref{eq:group}.
%For $I\in\cI(\cP)$ let $\tilde I:=\{\phi(I), \phi\in Aut(\cP)\}.$
Let $N(x)=|\{Tx, T\in G_\cP\}|,$ and note that $N(x)$ depends only on $\ideal x.$
From \eqref{eq:group} we obtain $N(x)=(q-1)^{|M(I)|}q^{|I\backslash M(I)|},$ where $I=\ideal x,$
and
  $$
    v_\alpha=N(x)|\tilde I|.
  $$
The eigenvalues are found from \eqref{eq:PQ} assuming that $\cA^\ast$ is realized
on the metric space $X^\bot=(\ff_q^n,d_{\cP^\bot}),$ i.e., that $\cP$ is self-dual. 
\begin{proposition}
Let $\alpha,\beta\in \cX$ be orbits. For $x\in \beta$ let $\tilde I$ be the 
orbit of $I=\ideal x$ under $Aut(\cP).$ Then
  \begin{equation}\label{eq:PP1}
   P_{\alpha\beta}=\sum_{I\in\tilde I: (I\cap J)\subset M(I)} (-1)^{|M(I)\cap J|}
   q^{|I\backslash M(I)|}(q-1)^{|M(I)\backslash(M(I)\cap J)|},
   \end{equation}
   where $J=\ideal\chi\in\cI(P^\bot)$ and $\chi\in N_\alpha^\ast$ is any fixed character. 
(Observe that the term $q^{|I\backslash M(I)|}$ is the same for all $I\in\tilde I$.)
\end{proposition}
\begin{proof} Let $y\in \ff_q^n$ be the image of $\chi$ under the isomorphism $X\cong X^\ast.$
Then we can think of $J$ as of the ideal $\ideal y _{\cP^\bot}.$  
Consider the sum \eqref{eq:PQ}:
  \begin{align}\label{eq:PP}
    P_{\alpha\beta}&=\sum_{x\in N_\beta} \chi(x)=\sum_{I\in\tilde I}\sum_{\begin{substack}{x: \ideal x=I\\x=(x_1,\dots,x_n)}
    \end{substack}}\prod_{i=1}^n \chi_{y_i}(x_i).
%    \sum_{x=(x_1,\dots,x_n)\in N_\beta} \prod_{i=1}^n \chi_{y_i}(x_i).
  \end{align}
Now observe that the terms with $x_i=0$ or $y_i=0$ contribute 1 in the product, 
so this sum is controlled by the intersection of the sets $I\cap J$ as subsets of $\n.$ 
Moreover, if $y_i\ne 0,$ then $\sum_{x_i\in \ff_q}\chi_{y_i}(x_i)=0.$ Thus, on account of \eqref{eq:group},
the nontrivial contribution to the sum on $x$ arises from the indices $i\in M(I)\cap J$ because in this
case $x\in \ff_q^\ast.$ Continuing from \eqref{eq:PP}, we obtain
   \begin{align}\label{eq:Pab}
   P_{\alpha\beta}&=\sum_{I\in\tilde I: (I\cap J)\subset M(I)}\prod_{i=1}^n \sum_{x_i} \chi_{y_i}(x_i).
  \end{align}
  If  $i\in (I\cap J),$ then $\sum_{x_i}\chi_{y_i}(x_i)=-1.$ If $i\not\in(I\cap J),$ then $y_i=0$ and
  $$
  \sum_{x_i}\chi_{y_i}(x_i)=
  \begin{cases} q-1 &\text{if }i\in  M(I)\backslash (I\cap J)\\
  q &\text{if }i\in I\backslash M(I).
  \end{cases}
  $$ 
  Substitution of these results into \eqref{eq:Pab} completes the proof.
\end{proof}  
   
   In the self-dual case, $Q=P$ and $m_\alpha=v_\alpha$ for all $\alpha.$ Generally, the first and second eigenvalues are connected by the
well-known expression \citep{bro89}, Lemma 2.2.1(iv)
  $$
  m_\alpha P_{\alpha\beta}=v_\beta Q_{\beta\alpha}.
  $$
Let $\cC$ and $\cC^\ast$ be a pair of dual codes. The inner distributions are defined as 
vectors $a=(a_\alpha,\alpha\in \cX), a_\alpha=|\cC\cap N_\alpha|$ for all $\alpha$, and 
$a'=(a'_\alpha,\alpha\in\cX^\ast), a'_\alpha=|\cC^\ast\cap N_\alpha^\ast|$ for all $\alpha'.$
They are related by the MacWilliams equations \citep{del73}
  \begin{equation}\label{eq:mw}
   a'=\frac1{|\cC|} aQ, \;a=\frac{|\cC|}{q^n}a'P.
  \end{equation}
Let us give one example.

\vspace*{.1in}{\em Example 4 (continued):} Consider again the NRT poset on $\n, n=mr.$ The orbit of a vector $x$ is formed
of all vectors with a fixed shape \eqref{eq:shape}, where $I=\ideal x.$ We have
   $$
   v_e=\binom{m}{e_1,e_2,\dots,e_r} (q-1)^{\sum_{i=1}^r (i-1)e_i}q^{\sum_{i=1}^r e_i},
   $$
where $\binom{m}{e_1,e_2,\dots,e_r}$ is the number of ways of choosing $e_i$ subsets of size $i=1,\dots,r$ out of
an $m$-set. The eigenvalues can be found from \eqref{eq:PP1} without difficulty. It is known \citep{bro89}, Lemma 2.2.1(iv) that the eigenvalues satisfy orthogonality
relations with weight $v_e.$ This enables us to interpret the eigenvalues of the ordered Hamming scheme
as $r$-variate orthogonal polynomials that belong to the family of multivariate Krawtchouk polynomials.
This approach is further developed in \citep{bar09b}.

\vspace*{.1in}Let us summarize our considerations of duality of codes and association schemes on poset metric spaces. 
First, there is no explicit need to realize $\cA^\ast$ on the dual poset.
At the same time, if this can be done (in the case of self-duality), then both the code and
its dual code can be visualized on the same Hasse diagram, which is convenient for their study.
This explains why many previous studies that involved MacWilliams theorems for poset metrics
\citep{gut98,mar99,dou02,bar09b}
dealt with self-dual posets. In such cases, the dual scheme $\cA^\ast$ is naturally identified 
with $\cA^\bot,$ which explains the switch of the ordering of the coordinates.

Next, if the poset is not self-dual, then duality of association schemes may have nothing
to do with linear-algebraic duality of codes. In this case we still can derive MacWilliams-like
relations between $\cC$ and $\cC^\bot$, but they do not fit the original context of association schemes
expressed by \eqref{eq:mw}. This approach is taken in a recent work by Choi et al.~\citep{choi12} whose main purpose
is to obtain such relations. Derivations in \citep{choi12} still rely on characters, but those
do not follow the structure of the dual scheme $\cA^\ast.$ 

Note also the case of the hierarchical poset (\citep{kim05a,kim05b,kim07}, \citep{pinheiro12}), 
which generally is not self-dual. The association schemes $\cA$ and $\cA^\perp$ are not 
isomorphic, but the inner
distributions nevertheless are compactly described in terms of the poset weight.

Our discussion can be summarized in the observation
that in the non-self-dual case, MacWilliams-type relations generally come in two different, inequivalent forms related to the association schemes $\cA^\ast$ and $\cA^\bot.$

\section{Shapes of codevectors}
MacWilliams relations are written with respect to the distribution of codevectors across the orbits of
the group of linear isometries of the space \eqref{eq:mw}.
For instance, for the Hamming metric, two vectors have the same weight if and only if
they are in the same orbit. In this context, the weight is a numerical invariant of vectors that characterizes the orbits.
Generally, we call such a numerical invariant the \emph{shape} of a vector.

\begin{definition}
Let $(\ff_q^n,d_\cP)$ be a poset metric space. A mapping $s:\ff_q^n\to \integers^m$ is called
a \emph{shape} mapping if it is constant on the orbits of $T\in GL_{\cP}(n)$.
%if and only if there is a linear isometry $T\in GL_{\cP}(n)$ such that $T(x)=y.$ 
The value that this mapping takes on the orbit of a vector $x\in \ff_q^n$ is called the shape of $x.$
\end{definition}

Shapes of vectors are known for only a small number of posets: hierarchical
posets (see \citep{felix}), in which case they are given by the poset weight, and NRT-posets; see \eqref{eq:shape}.
The structure of the isometry group \eqref{eq:group} suggests that shapes are determined by 
order ideals rather than vectors. This is the case for all the known examples.
In particular, if $\cP$ has the IE-property, then the shape mapping depends only on the isomorphism
class of ideals in the sense that $\shape(x)=\shape(y)$ if and only if $\ideal x\sim\ideal y.$
Generally, the shape is difficult to determine, and we believe that there is no possibility 
of finding a general expression for it.
Moreover, the shapes in the known examples not only determine the orbits, but also
other important invariants such as the weight of the vector and the packing
radius of the $1$-dimensional subspace generated by such vector. If shapes are
rare, then such useful shapes, that can determine other invariants, are much more so.

Clearly, in the general case the shape is not uniquely defined. Moreover, generally it is
difficult to check whether two vectors belong to the same orbit. We note that even in
simple cases of the Hamming space and the single chain the order of the isometry
group is exponential in $n.$ In these cases, as well as in the case of a regular tree
checking whether two vectors are isometric is easy. To be able to use shapes in the
study of structural properties of codes in a poset space, we generally would like to
be able to compute and compare shapes in time proportional to $\log|GL_{\cP}(n)|.$

\remove{
Clearly, the shape is not uniquely defined but we should explain what a reasonable shape should be in terms of algorithm complexity. Generally, direct (greedy) verification whether $x,y\in \ff_q^n$ belongs to the same orbit demands up to $ \left| GL_{\cP}(n) \right|$ checks: for each $T\in GL_{\cP}(n) $ one should check whether $T(x)=y$ or $T(x)\neq y$. Since in general the cardinality of $GL_{\cP}(n) $ grows exponentially with $n$ (for the Hamming case we have $ \left| GL_{\cP}(n) \right|=(q-1)^n n!$ and for a single chain we have $ \left| GL_{\cP}(n) \right|=(q-1)^nq^{n(n-1)/2}$ ), we wish that a suitable shape should demand much less check-operations. In all cases where a shape is available, those are very suitable: both for an hierarchical poset (where the shape is the weight of a vector) and for ordered Hamming orders (NRT spaces) the shape demands a number of check-operations that grows linearly with $n$. In general, if we denote by $DV(n)$ the (maximal) number of operations needed by the direct verification (greedy) method and by $SV(n)$ the (maximal) number of operations needed to compute and compare shapes, then we should demand, at least, that $\lim_{n\rightarrow \infty} SV(n)/DV(n) =0$ and at most that $\lim_{n\rightarrow \infty} \ln \left(DV(n)\right)/ SV(n) >0 $.}

\subsection{Codes on trees}
In this section we consider metric spaces $X=(\ff_q^n, \cP),$ where $\cP$ belongs to 
a class of posets whose Hasse diagrams are level-regular rooted
trees, introduced in Example 5 above. 

The group of linear isometries of $X$ is given in \eqref{eq:group}. It is a semidirect 
product of $Aut(\cP)$ and the matrix group $G_\cP,$
where $Aut(\cP)=S_{d_0}\times S_{d_0d_1}\times \dots\times S_{d_0d_1\dots d_{m-2}}.$
Proposition \ref{prop:lt} implies that shapes of vectors are determined by ideals, 
i.e., $x$ and $y$ are
in the same orbit of $GL_{\cP}(n)$ if and only if $\ideal x\sim\ideal y.$ Thus, orbits 
are characterized
by equivalence classes of ideals $\tilde I, I\in\cI(\cP).$ Ideals of $\cP$ themselves are
rooted trees (not necessarily regular), and their isomorphisms are obtained by 
restricting the isomorphisms of $\cP.$
We conclude that shapes of vectors in $X$ will be determined if we find suitable numerical 
invariants of rooted trees that encode them up to isomorphism.

Isomorphism of rooted trees is a classical problem in computer science \citep{Reed72},\citep{aho74}. 
There are many ways to encode a tree into a number so that it is possible
to decode that number to an isomorphic tree (a representative of the same
equivalence class). One way is as follows \citep{Reed72}. Let $H$ be a tree of height $h.$
Suppose all the vertices $i$ with $l(i)=k+1\le h$ are assigned labels, written as binary
strings. Vertex $j$ with $l(j)=k$ is given a label based on the labels of its sons.
Suppose that $j$ has $d_{k}$ sons whose labels are $C_1,C_2,\dots,C_k.$
These labels are sorted as binary numbers, so suppose that $C_1\le C_2\le \dots\le C_k,$
where identical labels are placed in arbitrary order.
Then $j$ is labeled as $(0|C_1|C_2|\dots|C_{d_k}|1).$ It remains to say what happens
if the vertex has no sons: in this case it is assigned the label $(01).$ 
The label assigned to the root is the label of the tree. It is a binary word that
can be also interpreted as an integer number. These labels play the role of shapes:
two ideals have the same label if and only if they are isomorphic.
The isometry group $GL_{\cP}(n)$ acts transitively on all the binary
vectors $x$ whose support ideals $\ideal x$ have the same label, and this label
therefore can be used as the shape of the codevector.

The tree is not self-dual, so there is no well-defined duality of linear codes on it.
At the same time, we can construct a self dual poset from a given poset $\cP$ by 
adjoining a mirror image to $\cP$ to $\cP.$ It is therefore of interest to study how the shapes behave under this and other similar operations on posets. We take up this question in the
next section.

\subsection{Operations on posets}
There are several standard ways to create new posets from old. In addition to
poset duality, the well-known operations include direct sums and direct products of posets
as well as ordinal sums and products \citep{Stanley2012}. Suppose that we start
with posets $\cP$ and $\cQ$ that have the IE property. We are interested whether
this property is inherited by posets that arise as a result of combining $\cP$ and $\cQ$,
and in the positive case, what is the shape of codevectors on $\cP\ast\cQ$ given the
shapes on $\cP$ and $\cQ$ (here $\ast$ is a generic notation for the operation).

Let $\cP=(\n,\preceq_{\cP})$ and $\cQ(\fs m,\preceq_{\cQ})$.
be two posets and let $\shape_{\cP}$ and $\shape_{\cQ}$ be their respective shape maps.
Throughout this section $\cP$ and $\cQ$ are assumed to have the IE property.

\subsubsection{Ordinal sum}
Define a poset $P\oplus Q$ on $\fs{m+n}$ in
the following way. Given $i,j\in\fs{m+n},$ put 
   $$
i\preceq_{\oplus}j\iff
 \left\{\begin{array}
[c]{c}
i,j\leq n\text{ and }i\preceq_{\cP}j\text{ or}\\
i,j>n\text{ and } (  i-n )  \preceq_{\cQ} (  j-n )  \text{
or}\\
i\leq n\leq j
\end{array} \right\}.
   $$
We note that the hierarchical poset is an ordinal sum of several antichains.
Also, ordinal sum gives us a way of making self-dual posets out of other posets.
It turns out that of the operations considered in this section,
it is only ordinal sum that allows the IE property to be carried to the resulting
poset from the component ones. 
\begin{lemma}
If both $\cP$ and $\cQ$ satisfy the IE-property, then so does $\cP\oplus\cQ.$
\end{lemma}

\begin{proof}
Given a subset $Y\subset\fs{n+m},$ consider the following sets:
\begin{align*}
    Y_{n}  &  =\{i\in\fs{m+n}: i\in Y\text{ and }i\le n\}\\
    Y^m &=\{i\in\fs{m+n}: i>n\}\\
   Y_m  &=Y^m-n:=\{i-n:i\in Y^m\}.
\end{align*}
Let us consider $I\in\mathcal{I} (  \cP\oplus \cQ ).$ It is immediate to
realize that $I^{m}\in\cI(\cQ)$ and $I_{n}\in \cI(\cP),$ and 
if $I_{m}\neq\emptyset$, then $I_{n}=\n$. Conversely, if
$\emptyset\neq I\in\cI(\cQ),$ then $\n \cup \{i+n\mid i\in I\}$ is
an ideal in $P\oplus Q,$ and if $I\in\cI(\cP),$ then, viewed as a subset of $\fs{n+m},$ 
it is also an ideal.

Assume that $I,J\in\cI(\cP\oplus \cQ)$ and suppose there is a poset 
isomorphism $\phi:I \to J$. Naturally, we have
that $\phi(I_{n})  =J_{n}$ and $\phi (I^{m} )=J^{m}$.

Suppose $I_{m}\neq\emptyset$ (and hence also $J_{m}\neq\emptyset$), 
then $I_{n}=J_{n}= [  n ]  $. Now note that $I_{m}$ and $J_{m}$ as
subsets of $\fs m$ are ideals in $\cQ$
and the map $\overline{\phi}:I_{m} \to J_{m}$ defined by $\overline
{\phi} (  i )  =\phi (  i+n )  -n$ is a poset isomorphism
between $I_{m}$ and $J_{m}$. By the IE-property of $\cQ$, there is $\xi\in
Aut (  Q )  $ such that $\xi (  I_{m} )  =J_{m}$. 
We define the map $\tilde{\phi}:\fs{n+m}\to \fs{n+m}$ as
follows:%
\begin{align*}
\tilde{\phi} (  i )   &  =\xi (  i-n )  +n\text{ for }i>n\\
\tilde{\phi} (  i )   &  =i\text{ for }i\leq n\text{.}%
\end{align*}
Than we have that $\tilde{\phi}\in Aut (  P\oplus Q )  $ and
$\tilde{\phi} (  I )  =J$.

Suppose now that $I_{m}=J_{m}=\emptyset$. Then we may view $I_{n}$ and
$J_{n}$ as isomorphic ideals in $\n\subset\fs{n+m}.$ 
By the extension property of $\cP$, there is $\xi\in Aut (\cP )$
such that $\xi(I_{n})=J_{n}$. Define the map 
$\tilde{\phi}:\fs{n+m}\to\fs{n+m}$ as follows:
\begin{align*}
\tilde{\phi} (  i )   &  =i\text{ for }i>n\\
\tilde{\phi} (  i )   &  =\xi (  i )  \text{ for }i\leq
n\text{.}%
\end{align*}
Then we have that $\tilde{\phi}\in Aut (  P\oplus Q )  $ and
$\tilde{\phi} (  I )  =J$.
\end{proof}

This lemma implies that, once shapes of vectors are defined both on $\cP$ and $\cQ,$
then in order to define a shape on $(\ff_q^{n+m},d_{\cP\oplus \cQ})$ 
it suffices to define it on $\cI(\cP\oplus\cQ).$ To simplify notation, we
denote a shape on $P\oplus Q$ -shape by $\mathrm{shape}_{\oplus}$. 
\begin{proposition} Let $I\in\cI(\cP\oplus\cQ).$ Then the following mapping
  $$
\mathrm{shape}_{\oplus} (  I )  =\begin{cases}
 (0,\mathrm{shape}_{\cQ}(I_m))  &\text{ if }
I_{m}\neq\emptyset\\
( 1,\mathrm{shape}_{\cP}(I_n))  &\text{ if }%
I_{m}=\emptyset
     \end{cases}
  $$
is a shape on $\cP\oplus \cQ$.
\end{proposition}
\begin{proof}
Let $I,J\in\mathcal{I} (  \cP\oplus \cQ )  $ be two ideals, 
and suppose that $\mathrm{shape}_{\oplus} (  I )  =\mathrm{shape}%
_{\oplus} (  J )  $. Suppose first that $I_{m}\neq\emptyset$, or
equivalently, that $\mathrm{shape}_{\oplus} (  I )  = (
0,\mathrm{shape}_{\cQ} (  I_{m} )   )  $. 
Then we must have $\shape_\oplus(J)=(0,\shape_\cQ(J_m))$
and thus $\mathrm{shape}_{\cQ}(I_m)=\shape_\cQ(J_m).$ It follows that
$I_{m}$ and $J_{m}$ are isomorphic as posets, let us say by an isomorphism $\phi$. Since
$I_{m}\neq\emptyset$ we have that $I_{n}=J_{n}=\n;$ hence
$I=\n \cup I^{m}$ and $J=\n \cup J^m.$ 
Define $\tilde{\phi}:I \to J$ by
   $$
   \tilde\phi(i)=\begin{cases} i&\text{if } i\le n\\
   \phi(i-n)+n &\text{if }i>n
 \end{cases}
  $$
and notice that it is an isomorphism of posets, i.e., $I\sim J$. 
Now suppose that $I_{m}=\emptyset.$ Then
  $$
   (1,\mathrm{shape}_{\cP}(I_n))=\mathrm{shape}_\cP(I)=\mathrm{shape}_\cP(J_n)=
(1,\mathrm{shape}_\cP(J_n)).
   $$
From the IE-property of $\cP$ we obtain that $I_{n}\sim J_{n},$
and since $I=I_{n}$, $J=J_{n},$ we conclude that $I\sim J$. 

Conversely, let us assume that $I\sim J$. Then we must have
$I_{n}\sim J_{n}$ and $I^{m}\sim J^{m}$. Suppose that $I_{m}=\emptyset$. The
IE-property ensures that $\mathrm{shape}_{\cP} (  I_{n} )
=\mathrm{shape}_{\cP} (  J_{n} ),$ and hence
  $$
\mathrm{shape}_{\oplus} (  I )  = (  1,\mathrm{shape}_{P} (
I_{n} )   )  = (  1,\mathrm{shape}_{P} (J_{n} )
 )  =\mathrm{shape}_{\oplus} (  J ).
   $$
Now suppose that $I_{m}\neq\emptyset,$ then
   $$
I=\n  \cup I^{m},\;J=\n  \cup J^{m}%
   $$
and the isomorphism $\phi:I \to J$ maps $I^{m}$ to $J^{m}$. It follows
that the map $\tilde{\phi}:I_{m} \to J_{m}$ defined by $\tilde{\phi
} (  i )  =\phi (  i+n )  -n$ is a poset isomorphism and by
the IE-property we have that $\mathrm{shape}_{Q} (  I_{m} )
=\mathrm{shape}_{Q} (  I_{m} ),$ so that
\[
\mathrm{shape}_{\oplus} (  I )  = (  0,\mathrm{shape}_{Q} (
I_{m} )   )  = (  0,\mathrm{shape}_{Q} (  J_{m} )
 )  =\mathrm{shape}_{\oplus} (  J ).
\]
\end{proof}

\subsubsection{Direct sum}
Given $\cP$ and $\cQ,$ the direct sum operation results in a poset $\cP+\cQ$
in which the diagrams of $\cP$ and $\cQ$ are drawn  ``side by side.''
Namely, given $i,j\in \fs{n+m}$ we have
   $$
i\preceq_{\oplus}j\iff%
\left\{\begin{array}
[c]{c}%
i,j\leq n\text{ and }i\preceq_{\cP}j\text{ or}\\
i,j>n\text{ and } (  i-n )  \preceq_{\cQ} (  j-n )
\end{array}\right\}.
   $$
This poset does not inherits the IE property from $\cP$ and $\cQ$. To give 
a simple example, suppose that $\cP$ and $\cQ$ are not isomorphic. Then we pick
a $\cP$-minimal element $i\in\n$ and a $\cQ$-minimal element
$j\in\fs m$. The sets $\{i\}$ and $\{j+n\}$ are trivially isomorphic as ideals 
in $P\oplus Q,$ but there is no poset automorphism on 
$\cP\oplus \cQ$ that would exchange them. 

\subsubsection{Ordinal product}

Given posets $\cP= (\n,\preceq_{\cP} )  $ and
   $
   \cQ= ( \fs m ,\preceq_{\cQ} )  $, the poset $\cP\otimes
   \cQ= (\n\times\fs m  ,\preceq_{\otimes} )
   $ 
is defined by the relations
   $$
 (  i,j )  \preceq_{\otimes} (  i^{\prime},j^{\prime} )
\iff i=i^{\prime}\text{ and }j\preceq_{Q}j^{\prime}\text{.}%
   $$
The IE property is not inherited from $\cP$ and $\cQ$. Indeed,
suppose $\cP$ is a poset on $\{1,2\}$ defined by the
relations $i\preceq_{\cP}i,i=1,2.$ Suppose that $\cQ$ is a hierarchical poset
having at least two minimal elements, let us say $i$ and $j,$ 
and assume that they are not maximal, i.e., there is $k\in \fs m$ such that
$i\preceq_{\cQ}k$ and $j\preceq_{\cQ}k$. Then the sets 
$ \{(1,i ),(1,j )\}$ and $\{(1,i),(2,i)\}$ are ideals 
(since $1,2,i,j$ are all minimal elements), 
and since neither $i$ and $j$ nor $1$ and $2$ are comparable, those
ideals are isomorphic. At the same time, 
there is no $\phi\in Aut(\cP\otimes \cQ )$ that takes 
$\{(1,i),(1,j)\}$ to $\{(1,i),(2,i)\}.$ Indeed, $(1,i)\preceq_{\otimes} (1,k)$ 
and $(1,j)\preceq_{\otimes}(1,k),$ but there is no $(x,y)$ such that 
$(1,i)\preceq_{\otimes} (x,y)$ and $(2,i)\preceq_{\otimes}(x,y)$
since $(1,i)\preceq_{\otimes} (x,y)$ implies that $x=1$. 

\subsubsection{Direct product}

Given posets $\cP= (\n,\preceq_{\cP} )  $ and
$\cQ= (\fs m,\preceq_{\cQ} )  $, the poset $\cP\times
\cQ= (\n \times \fs m  ,\preceq_{\times} )$ is defined by the relation
\[
 (  i,j )  \preceq_{\times} (  i^{\prime},j^{\prime} )  \iff
i\preceq_{\cP}i^{\prime}\text{ and }j\preceq_{\cQ}j^{\prime}\text{.}%
\]
It is easy to see that a subset $I\times J\subset \n\times \fs m$
is an ideal of $\cP\times \cQ$ if and only if $I\in\cI(\cP)$ and $J\in\cI(\cQ).$
Moreover, $\phi= (  \phi_{\cP},\phi_{\cQ} )  \in Aut ( \cP\times\cQ)$ 
if $\phi_{\cP}\in Aut ( \cP )  $ and $\phi_{\cQ}\in Aut (\cQ )  $. 
At the same time, clearly not any $\phi\in Aut (\cP\times\cQ)$ can be expressed 
in such a way. It follows that the extension property does not necessarily 
hold on $\cP\times \cQ$, as can be seen in the following example.

\vspace*{.1in}
{\em Example:}
Let $\cP= (\fs2  ,\preceq_{\cP})$ and
$Q= (\fs3  ,\preceq_{\cQ} )  $ be defined by the
relations $1\prec_{\cP}2,1\prec_{\cQ}2,$  and $1\prec_{\cQ}3$. Then $\cP\times \cQ$
is generated by the relations
\begin{align*}
 (  1,1 )& \preceq_{\times} (  2,1 )  ;\; (
1,1 )  \preceq_{\times} (  1,2 )  ; (  1,1 )
\preceq_{\times} (  1,3 )  \\
 (  2,1 )    & \preceq_{\times} (  2,2 )  ; (
1,2 )  \preceq_{\times} (  2,2 )  ; (  1,2 )
\preceq_{\times} (  2,3 )  ; (  1,3 )  \preceq_{\times
} (  2,3 )  \text{.}
\end{align*}
The sets
\[
I= \{   (  1,1 )  , (  1,2 )   \}  ,J= \{
 (  1,1 )  , (  2,1 )   \}
\]
are both  ideals in $\cP\times \cQ$ that are isomorphic as posets, but there
is no $\phi\in Aut (\cP\times \cQ )  $ such that $\phi (I)=J.$

\vspace*{.2in}  
Some of the initial ideas of this paper appeared earlier in the 
extended abstract \citep{Barg12}. Here they are both developed and refined in 
a number of ways.

%% The Appendices part is started with the command \appendix;
%% appendix sections are then done as normal sections
%% \appendix

%% \section{}
%% \label{}

%% References
%%
%% Following citation commands can be used in the body text:
%% Usage of \cite is as follows:
%%   \cite{key}         ==>>  [#]
%%   \cite[chap. 2]{key} ==>> [#, chap. 2]
%%

%% References with bibTeX database:

%\bibliographystyle{elsarticle-num}
%\bibliography{<your-bib-database>}

%% Authors are advised to submit their bibtex database files. They are
%% requested to list a bibtex style file in the manuscript if they do
%% not want to use elsarticle-num.bst.

%% References without bibTeX database:

% \begin{thebibliography}{00}

%% \bibitem must have the following form:
%%   \bibitem{key}...
%%

% \bibitem{}

% \end{thebibliography}

\providecommand{\bysame}{\leavevmode\hbox to3em{\hrulefill}\thinspace}
\providecommand{\href}[2]{#2}

\end{document}